\documentclass[conference, 10pt, a4paper]{IEEEtran}
\ifCLASSINFOpdf
\else
\fi
\usepackage{rayyli-short}
\usepackage{nohyperref}
\begin{document}
%
\title{Efficiently decodable insertion/deletion codes for high-noise and high-rate regimes}

\author{\IEEEauthorblockN{Venkatesan Guruswami}
\IEEEauthorblockA{
Carnegie Mellon University\\
Pittsburgh, PA 15143\\
Email: guruswami@cmu.edu}
\and
\IEEEauthorblockN{Ray Li}
\IEEEauthorblockA{
Carnegie Mellon University\\
Pittsburgh, PA 15143\\
Email: ryli@andrew.cmu.edu}
}


%


\maketitle

\begin{abstract}
This work constructs codes that are efficiently decodable 
from a constant fraction of \emph{worst-case} insertion and deletion errors in three parameter settings:
(i) Binary codes with rate approaching 1; (ii) Codes with constant rate for error fraction approaching 1 over fixed alphabet size; and (iii)
Constant rate codes over an alphabet of size $k$ for error fraction approaching $(k-1)/(k+1)$. When errors are constrained to deletions alone,
efficiently decodable codes in each of these regimes were constructed recently.
We complete the picture by constructing similar codes that are
efficiently decodable in the insertion/deletion regime.
\end{abstract}


%

\section{Introduction}

\newcommand\equivlemmatext{
  Let $C\subseteq [k]^n$ be a code, and let $t<n$ be a positive integer.
  The following are equivalent.
  1) $C$ is decodable under up to $t$ insertions.
  2) $C$ is decodable under up to $t$ deletions.
  3) $C$ is decodable under up to $t$ insertions and deletions.
}

This work addresses the problem of constructing error-correcting codes
that can be efficiently decoded from a constant fraction of
\emph{worst-case} insertions and deletions. The main results
generalize analogous results for the situation when errors are
restricted to only deletions.

Coding for a constant fraction of adversarial insertions and deletions has been considered previously
by Schulman and Zuckerman \cite{SZ1999}.
They construct constant-rate binary codes that are efficiently decodable from a small
constant fraction of worst-case insertions and deletions and can also handle a small
fraction of transpositions.


Our work primarily builds off recent results by Guruswami, Bukh, and Wang
\cite{BG2015, GW2015}, which address the construction and efficient decoding
of codes for constant fractions of deletions.
These works establish three results, providing families of codes
with each of the following parameters.
\begin{enumerate}
  \item Families with rate approaching 1 decoding a constant fraction of deletions
  \item Families with constant rate decoding a fraction of deletions approaching 1
  \item Families over a fixed alphabet of size $k$ with constant rate
  and decoding a fraction of deletions approaching $1-\frac{2}{k+1}$ (In particular, one gets binary codes for correcting a deletion fraction approaching $1/3$.)
\end{enumerate}
Over an alphabet of size $k$, it is impossible to have a constant
rate code that corrects a $1-\frac{1}{k}$ fraction of deletions.
The last result establishes that the maximum correctable fraction of
deletions of a constant rate code is $1-\Theta(\frac{1}{k})$.



Combinatorially, decoding a given number of worst-case insertions and deletions
is identical to decoding the same number of worst-case deletions.
This is established in the following lemma,
originally given by Levenshtein \cite{levenshtein1966}.
\begin{lemma}
  \label{equivlemma}
  \equivlemmatext
\end{lemma}

Lemma \ref{equivlemma} establishes that
the codes provided in the three constructions must also be
capable of decoding both insertions and deletions.
The task that remains, and which our work addresses, is to construct
codes in the \emph{same parameter settings} that can efficiently
correct a combination of insertions \emph{and} deletions.

The regime under which errors are insertions and deletions is closely related to
\textit{edit-distance} (also known as Levenshtein distance),
which measures errors of a code under insertions, deletions, and substitutions.
A substitution can be viewed as a deletion
followed by an insertion.
Thus, all results established
in the insertion and deletion regime, both constructive and algorithmic,
hold in the edit-distance regime when the number of errors is cut in half, and therefore in the traditional coding theory setting
in which the only errors are substitutions. The edit-distance is a more challenging model, however; while the  Gilbert-Varshamov bound
gives codes over size $k$ alphabets
that can correct up to a fraction of substitutions approaching $\frac{1}{2}(1-\frac{1}{k})$, 
the question of whether there exist positive rate codes capable of correcting a deletion fraction approaching
$1-\frac{1}{k}$ is still open.

\subsection{Prior Results}

These are the efficiently decodable code constructions in the deletion-only regime that
we are generalizing to the insertion/deletion regime.

1) A binary code family of rate $1-\tilde O(\sqrt\eps)$
  that can be efficiently decoded from
  an $\eps$ fraction of worst-case deletions, for all $\eps$ smaller than some absolute constant $\eps_0 > 0$.  Furthermore, the codes are constructible, encodable, and decodable,
  in time $N^{\poly(1/\eps)}$, where $N$ is the block length. [Theorem 4.1 from \cite{GW2015}]
 
2) For any $\eps > 0$, a code family over an alphabet of size 
 $\poly(1/\eps)$ and rate $\Omega(\eps^2)$ that can be decoded from a $1-\eps$ fraction of worst-case deletions.
  Furthermore, this code is constructible, encodable, and decodable in time
  $N^{\poly(1/\eps)}$. [Theorem 3.1 from \cite{GW2015}]
  
3) For all integers $k\ge2$ and all $\eps > 0$, a code family over alphabet size $k$ of positive rate $r(k,\eps) > 0$ 
  that can be  decoded from
  a $1-\frac{2}{k+1}-\eps$ fraction of worst-case deletions in $O_{k,\eps}(N^3(\log N)^{O(1)})$ time.

\subsection{Our Results and Techniques}

\newcommand\higherroronetext{
  Fix an integer $k\ge 2$ and $\eps>0$.
  For infinitely many and sufficiently large $N$,
  there is an explicit code $C\subseteq\{0,1\}^N$
  with rate $r(k,\eps)=(\eps/k)^{O(\eps^{-12})}$
  over a size $k$ alphabet
  that can be decoded from a
  $1-\frac{2}{k+1}-\eps$ fraction of worst-case insertions and deletions
  in time $O_{k,\eps}(N^3\polylog (N))$.
  Furthermore, this code is constructible in time $O_{k,\eps}(N\log^2N)$.
}
\newcommand\higherrortwotext{
  For any $\eps>0$, there exists a family of
  codes over an alphabet
  of size $\poly(1/\eps)$ and rate $\Omega(\eps^5)$
  that can be efficiently decoded from a $1-\eps$ fraction
  of insertions/deletions.
  Furthermore, this code is constructible, encodable, and decodable in time
  $N^{\poly(1/\eps)}$.
}
\newcommand\highratetext{
  There exists a constant $\eps_0>0$ such that the following holds.
  Let $0<\eps<\eps_0$. There is an explicit binary code $C\subseteq \{0,1\}^N$
  with rate $1-\tilde O(\sqrt{\eps})$
  that is decodable from an $\eps$
  fraction of insertions/deletions
  in $N^{\poly(1/\eps)}$ time.
  Furthermore, $C$ can be constructed and encoded
  in time $N^{\poly(1/\eps)}$.
}

Our work constructs the following three families of codes.
\begin{enumerate}
  \item
  alphabet size: 2,
  rate: $1-\tilde O(\sqrt\eps)$,
  insertion/deletion fraction: $\eps$,
  decoding time: $N^{\poly(1/\eps)}$. (Thm.~\ref{highratetheorem})
  \item
  alphabet size: $\poly(1/\eps)$,
  rate: $\Omega(\eps^5)$,
  insertion/deletion fraction: $1-\eps$,
  decoding time: $N^{\poly(1/\eps)}$. (Thm.~\ref{higherrortwotheorem})
  \item
  alphabet size: $k\ge 2$,
  rate: $(\eps/k)^{\mathrm{poly}(1/\eps)}$,
  insertion/deletion fraction: $1-\frac{2}{k+1}-\eps$,
  decoding time: $O_{k,\eps}(N^3\poly\log(N))$.  (Thm.~\ref{higherroronetheorem})
\end{enumerate}

\begin{remark*}
  Theorem \ref{higherrortwotheorem} gives constant rate codes that decode
  from a $1-\eps$ fraction of insertions/deletions.
  This also follows as a corollary from Theorem \ref{higherroronetheorem}.
  However, the rate of the construction in
  Theorem \ref{higherroronetheorem}
  is $(\eps/k)^{\mathrm{poly}(1/\eps)}$, which is far worse than $\poly(\eps)$. The main point of 
\ref{higherroronetheorem} is to highlight the near-tight trade-off between alphabet size and insertion/deletion fraction.
\end{remark*}
\begin{remark*}
  At the expense of slightly worse parameters,
  the construction and decoding complexities
  in Theorems \ref{highratetheorem} and \ref{higherrortwotheorem}
  can be improved to $\poly(N)\cdot(\log N)^{\poly(1/\eps)}$.
  See Theorems \ref{highratefasttheorem} and
  \ref{higherrortwofasttheorem}.
\end{remark*}

Theorems \ref{higherrortwotheorem} and \ref{higherroronetheorem}
use the powerful idea of list decoding, exemplified
in \cite{BG2015}. A normal decoding algorithm
is required to return the exact codeword, but a list decoding algorithm
is allowed to return a list of codewords containing the correct codeword.
The codes for both theorems are decoded by first applying
a list decoding algorithm, and then noting that
if the easier list decoding is gauranteed to succeed
(that is, returns a list containing the correct codeword),
one can simply pass through the resulting list and choose the unique codeword
that has sufficiently small distance from the received word.
The codeword will be unique because the codes constructed
are provably decodable under the required number of insertion/deletions
according to the results in \cite{BG2015, GW2015}.

The extent of difference between the insertion/deletion
decoding algorithms and their deletion-only analogues
varies depending on the parameter setting.
For a $1-\frac{2}{k+1}-\eps$ fraction of insertions/deletions, the decoding algorithm
uses the same list decoding approach
as the deletion-only decoding algorithm in \cite{BG2015}.
For a $1-\eps$ fraction of insertions/deletions, we adopt
the list decoding approach that in fact simplifies the construction presented
in \cite{GW2015}.
For achieving a rate of $1-\eps$, we use the same code as in
\cite{GW2015} with different parameters, but considerably more
bookkeeping is done to provide a provably correct
decoding algorithm. In particular, both
Theorem 4.1 from \cite{GW2015} and
Theorem \ref{highratetheorem}
place chunks of 0s between inner codewords.
However, while identifying buffers
in the received word in the deletion-only case
merely requires identifying long runs of 0s,
identifying buffers in the insertion/deletion case
requires identifying strings of fixed length with sufficiently
small fraction of 1s.

\section{Preliminaries}
Let $[n]$ denote the set $\{1,\dots,n\}$.
For a string $s$, let $|s|$ denote the length of the string.
Define $\Delta(c,c')$ to be the \textit{insertion/deletion distance}
between $c$ and $c'$, that is, the number of insertions/deletions
needed to manipulate $c$ into $c'$.
For two words $c,c'\in C$, let $\LCS(c,c')$ be
the length of the longest common subsequence of $c$ and $c'$. Define
$\LCS(C)= \max_{c,c'\in C, c\neq c'} \LCS(c,c')$.
For the same reason that Lemma \ref{equivlemma} is true,
we have $\Delta(c,c') = |c|+|c'|-2\LCS(c,c')$.

A code $C$ of block length $n$ over an alphabet
$\Sigma$ is a subset $C\subseteq \Sigma^n$. The rate of $C$
is defined to be $\frac{\log |C|}{n\log|\Sigma|}$.
The encoding function of a code is a map $\Enc:[|C|]\to\Sigma^n$
whose image equals $C$
(with messages identified with $[|C|]$ in some canonical way),
and the decoding function of a code is a map $\Dec:\Sigma^*\to C$.

A code is encodable in time $f(n)$ if, for all elements of $[|C|]$,
the map $\Enc$ can be computed in time $f(n)$.
A code is decodable from $t$ (or, a $\delta$ fraction of)
worst-case insertions and
deletions in time $f(n)$ if, for all $c\in C$, and
for all $s$ such that $\Delta(s, c)\le t$ (or $\delta n$),
$\Dec(s)$ can be computed in time $f(n)$ and evaluates to $c$.
A code is constructible in time $f(n)$ if descriptions of $C, \Dec,$
and $\Enc$ can be produced in time $f(n)$.

Just as in \cite{SZ1999,BG2015, GW2015},
our constructions use the idea of code concatenation:
If $C_{\text{out}}\subseteq \Sigma_{\text{out}}^n$
is an ``outer code'' with encoding function $\Enc_{\text{out}}$,
and $C_{\text{in}}\subseteq \Sigma_{\text{in}}^m$
is an ``inner code'' with encoding function
$\Enc_{\text{in}}:\Sigma_{\text{out}}\to\Sigma_{\text{in}}^m$,
then the concatenated code
$C_{\text{out}}\circ C_{\text{in}}\subseteq \Sigma_{\text{in}}^{nm}$
is a code whose encoding function first applied $\Enc_{\text{out}}$ to the message,
and then applied $\Enc_{\text{in}}$ to each symbol of the resulting outer codeword.

Thoughout the paper,
$c,c'$ denote codewords,
$s,s'$ denote codewords modified under insertions and deletions,
and $w,w'$ denote inner codewords of concatenated codes.
We let $n$ denote the block length of the code, unless we
deal with a concatenated code, in which case $n$ denotes
the block length of the outer code, $m$ denotes the block length of the
inner code, and $N=nm$ denotes the block length of the entire code.
Alphabet sizes are denoted by $k$, and field sizes for outer Reed Solomon
codes are denoted by $q$.

\section{High Rate}
  \begin{lemma}[Proposition 2.5 of \cite{GW2015}]
    \label{denseinnercodelemma}
    Let $\delta,\beta\in(0,1)$. Then,
    for every $m$, there exists a code $C\subseteq\{0,1\}^m$
    of rate $R=1-2h(\delta)-O(\log(\delta m)/m)-2^{-\Omega(\beta m)/m}$
    such that
    \begin{itemize}
      \item for every string $s\in C$, every interval of length $\beta m$
      in $s$, contains at least $\beta m/10$ 1's,
      \item $C$ can be corrected from a $\delta$ fraction of worst-case
      deletions, and
      \item $C$ can be found, encoded, and decoded in time $2^{O(m)}$.
    \end{itemize}
  \end{lemma}
  \begin{theorem}
    \label{highratetheorem}
    \highratetext
  \end{theorem}
  \begin{proof} 
    With hindsight, let $\eps_0=\frac{1}{121^2}$, and let $0<\eps<\eps_0$.
    Consider the concatenated construction with the outer code being a
    Reed-Solomon code that can correct a $60\sqrt\eps$
    fraction of errors and erasures.
        For each $1\le i\le n$, we replace the $i$th coordinate
    $c_i$ with the pair $(i,c_i)$; to ensure that this doesn't affect the rate much, we take the RS code  to be over $\FF_{q^h}$,
    where $n=q$ is the block length and $h=1/\eps$.
    We encode each outer symbol pair in the inner code,
    defined as follows.

    The inner code is a good binary insertion/deletion code $C_1$
    of block length $m$ decoding a $\delta=40\sqrt\eps<\half$
    fraction of insertions and deletions, such that
    every interval of length $\delta m/16$ in a codeword
    has at least $1/10$ fraction of 1s.
    This code can be found using Lemma \ref{denseinnercodelemma}.
    We also assume each codeword begins and ends with a 1.
    
    Now take our concatenated Reed-Solomon code of block length $mn$,
    and between each pair of adjacent inner codewords of $C_1$,
    insert a \textit{chunk} of $\delta m$ 0s. This gives us our final code $C$
    with block length $N=nm(1+\delta)$.
    
    \begin{lemma}
      The rate of $C$ is $1-\tilde O(\sqrt\eps)$. 
    \end{lemma}
    \begin{proof}
      The rate of the outer RS code
      is $(1-120\sqrt\eps)\frac{h}{h+1}$,
      and the rate of the inner
      code can be taken to be $1-2h(\delta)-o(1)$
      by Lemma \ref{denseinnercodelemma}.
      Adding in the buffers reduces the rate by a factor of $\frac{1}{1+\delta}$.
      Combining these with our choice of $\delta$ gives us a total
      rate for $C$ of $1-\tilde O(\sqrt{\eps})$.
    \end{proof}

    \begin{lemma}
      The code $C$ can be decoded from an $\eps$ fraction of
      insertions and deletions in time $N^{\poly(1/\eps)}$.
    \end{lemma}
    Consider the following algorithm that runs in time
    $N^{\poly(1/\eps)}$ for decoding the received word:

1) Scan from the left of the received word. Every time we encounter
      a substring of length exactly $\delta m$ with at most
      $\frac{1}{160}$ fraction of 1s (or $\delta m/160$ 1s),
      mark it as a \textit{decoding buffer}.
      Then, continue scanning from the end of the buffer and repeat.
      This guarantees no two buffers overlap.
      This takes time $\poly(N)$.
      
\smallskip
  2) Start with an empty set $L$.
      The buffers divide the received word into strings
      which we call \textit{decoding windows}.
      For each decoding window, apply the decoder
      from Lemma \ref{denseinnercodelemma} to recover
      a pair $(i,r_i)$. If we succeed, add this pair to $L$.
      This takes $N^{\mathrm{poly}(1/\eps)}$ time.

\smallskip
3) If for any $i$, $L$ contains multiple pairs with
      first coordinate $i$, remove all such pairs
      from $L$. $L$ thus contains at most one pair
      $(i,r_i)$ for each index $i$. Then apply the RS
      decoding algorithm to the string $r$ whose
      $i$th coordinate is $r_i$ if $(i,r_i)\in L$ and erased
      otherwise. This takes time $\poly(N)$.

    \begin{remark*}
      In the deletion only case, the decoding buffers
      are runs of at least $\delta m/2$ contiguous zeros.
      Runs of consecutive zeros are obviously a poor choice for decoding buffers in the presence of insertions,
      as we can destroy
      any buffer with a constant number of insertions.
    \end{remark*}

    Note that the total number of insertions/deletions we can make is at most
    $(1+\delta)mn\eps <2\eps mn$.

    Suppose our received codeword is
    $s=u_1\circ y_1\circ u_2\circ \cdots\circ u_{n'}$,
    where $y_1,\dots, y_{n'-1}$ are the identified
    decoding buffers and $u_1,\dots, u_{n'}$ are the decoding windows.
    Then consider a canonical mapping from characters of $c$ to
    characters of $s$
    where $u_i$ is mapped to by a substring $t_i$
    of $c$, $y_i$ is mapped to by a string $x_i$, so that
    $c = t_1\circ x_1\circ\cdots\circ t_{n'}$
    and
   $ \Delta(c,s)
        = \sum_{i=1}^{n'}\Delta(u_i,t_i)
        + \sum_{i=1}^{n'-1}\Delta(y_i,x_i)$.

    With our canonical mapping, we can identify $\LCS(c,s)$ many characters in $s$ with characters in $c$.
    Intuitively, these are the characters that are uncorrupted
    when we transform $c$ into $s$ using insertions and deletions.
    Call a received buffer $y_i$ in $s$ a \textit{good decoding buffer}
    (or \textit{good buffer} for short) if
    at least $\frac{3}{4}\delta m$ of its characters are identified with
    characters from a single chunk of $\delta m$ 0s in $c$. Call a decoding buffer
    \textit{bad} otherwise.
    Call a chunk of $\delta m$ 0s in $c$ \textit{good} if at least
    $\frac34 \delta m$ of its zeros map to characters in
    single decoding buffer. Note that there is a natural
    bijection between good chunks in $c$ and good decoding
    buffers in $s$.
    
    \begin{lemma}
      \label{bad-buffer-lemma}
      The number of bad decoding buffers of $s$ is at most $8\sqrt{\eps}n$.
    \end{lemma}
    \begin{proof}
      Suppose we have a bad buffer $y_i$. It either contains characters
      from at least two different chunks of $\delta m$ 0s in $c$ or contains
      at most $\frac{3\delta m}{4}$ characters from a single chunk.
      
      In the first case, $x_i$ must contain characters in two
      different chunks so its length must be at least $m$,
      so $y_i$ must have been obtained from at least
      $m-\delta m>\delta m > 40\sqrt\eps m$ deletions from $x_i$.
      
      In the second case, if $x_i$ has length
      at most $\frac{7\delta m}{8}$
      then the insertion/deletion distance between $x_i$ and $y_i$ is at least
      $\frac{\delta m}{8} = 5\sqrt{\eps}m$.
      Otherwise, $x_i$ has at least $\frac{\delta m}{8}$ charaters
      in some inner codeword of $c$,
      so $x_i$ has at least $\frac{\delta m}{80}$ 1s,
      so we need at least
      $\frac{\delta m}{80} - \frac{\delta m}{160} = \frac14\sqrt\eps m$
      deletions to obtain $y_i$ from $x_i$.

      By a simple counting argument,
      the total number of bad buffers we can have is
      at most $\frac{2\eps mn}{\frac14 \sqrt\eps m} = 8\sqrt\eps n$. 
    \end{proof}

    \begin{lemma}
      \label{good-buffer-lemma}
      The number of good decoding buffers of $s$
      is at least $(1-8\sqrt{\eps})n$.
    \end{lemma}
    \begin{proof}
      It suffices to prove the number of good chunks of $c$ is at least
      $(1-8\sqrt\eps)n$.
      If a chunk is not mapped to a good buffer, at least
      one of the following is true.
      \begin{enumerate}
        \item The chunk is ``deleted'' by inserting enough 1s.
        \item Part of the chunk is mapped to a bad buffer that
        contains characters from $t-1\ge 1$ other chunks.
        \item Part of the chunk is mapped to a bad buffer
        that contains no characters from other chunks.
      \end{enumerate}
      In the first case, we need at least
      $\frac{\delta m}{160}=\frac14\sqrt\eps m$ insertions to delete
      the chunk.
      In the second case, creating the bad buffer costs at least
      $(t-1)(m-\delta m)\ge \frac{t\delta m}{2}$ deletions,
      which is at least $20\sqrt\eps m$ deletions per chunk.
      In the third case, creating the bad buffer costs at least
      $\frac14\sqrt\eps m$ edits by the argument
      in Lemma \ref{bad-buffer-lemma}.
      Thus, we have
      at most $\frac{2\eps mn}{\frac14 \sqrt\eps m} = 8\sqrt\eps n$
      bad chunks, so we have at least $(1-8\sqrt\eps)n$ good chunks,
      as desired.
    \end{proof}

    Since there are at least $(1-8\sqrt\eps)n$ good decoding buffers
    and at most $8\sqrt\eps n$ bad decoding buffers,
    there must be at least $(1-16\sqrt\eps)n$ pairs of consecutive
    good decoding buffers. For any pair of consecutive good decoding
    buffers $y_{j-1},y_j$ in $s$, the corresponding two good chunks of
    $\delta m$ 0s in $c$ are consecutive unless there is at least one
    bad chunk in between the two good chunks, which happens
    for at most $8\sqrt\eps n$ pairs. Thus,
    there are at least $(1-24\sqrt\eps)n$ pairs of consecutive
    good decoding buffers in $s$
    such that the corresponding good chunks of 0s in $c$
    are also consecutive.

    Now suppose $w$ is an inner codeword between two good chunks
    with corresponding consecutive good decoding buffers, $y_{j-1},y_j$.
    The corresponding decoding window between the decoding buffers
    is $u_j$, mapped to from $t_j$, a substring of $c$. We claim
    that most such $w$ are decoded correctly.

    For all but
    $2\frac{2\eps mn}{\delta m/8} + 2\frac{2\eps mn}{\delta m/8}
    + \frac{2\eps mn}{\delta m/4} < 2\sqrt\eps n$ 
    choices of $j$, we have
    $\Delta(x_{j-1},y_{j-1})\le \frac{\delta m}{8}$,
    $\Delta(x_{j},y_{j})\le \frac{\delta m}{8}$,
    and $\Delta(t_{j},u_{j})\le \frac{\delta m}{4}$.
    When we have an inner codeword $w$ and an index $j$ such that
    all these are true, we have
    $|x_{j-1}|, |x_j|\le \frac{9\delta m}{8}$,
    and each of $x_{j-1},x_j$ shares at least
    $\frac{3\delta m}{4}$ characters
    with one of the chunks of $\delta m$ 0s neighboring
    $w$. It follows that $x_{j-1},x_j$ each contain at most
    $\frac{3\delta m}{8}$ characters of $w$. Additionally,
    by the definition of a good chunk, $u_j$ contains
    at most $\frac{\delta m}{4}$ characters in each of the
    chunks neighboring $w$.
    Thus, we have $\Delta(w,t_j)\le \frac{3\delta m}{4}$,
    in which case, 
    $\Delta(w,t_j)\le \Delta(w,t_j)+\Delta(t_j,u_j)\le \delta m$.
    Thus, for at least $(1-24\sqrt\eps)n - 2\sqrt\eps n = (1-26\sqrt\eps)n$
    inner words $w$, there exists $j\in\{1,\dots, n'\}$
    such that $\Delta(w,u_j)\le \delta m$.

    Therefore, our algorithm detects at least $(1-26\sqrt\eps)n$ correct
    pairs $(i,r_i)$.
    Since our algorithm detects at most $(1+8\sqrt\eps)n$
    pairs total, we have at most $34\sqrt\eps n$ incorrect pairs.
    Thus, after removing conflicts, we have at least $(1-60\sqrt\eps)n$
    correct values, so our Reed Solomon decoder will succeed.
  \end{proof} 

\vspace{-2ex}
  \begin{remark*}
    Our decoding algorithm succeeds as long as the
    inner code can correct up to a $\delta$ fraction of insertions/deletions
    and consists of codewords such that every interval of length
    $\delta m/16$ has at least $1/10$ fraction of 1s.
The time complexity of
    Theorem \ref{highratetheorem} can be improved using a more
    efficient inner code, at the cost of reduction in rate.

    Because of the addition of buffers, the code of Theorem \ref{highratetheorem}
    may not be dense enough to use as an inner code.
    The inner code needs to have $1/10$ fraction of 1s 
    for every interval of length $\delta m/16$.
    However, we can modify the construction of the inner concatenated code so that
    the inner codewords of the inner code in Theorem \ref{highratetheorem}
    have at least $1/5$ fraction of
    1s in every interval of length $\delta m/16$.
    This guarantees that the inner codewords of our two level construction
    have sufficiently high densities of 1s.
    This is summarized in the following theorem.
  \end{remark*}
  \begin{theorem}
    \label{highratefasttheorem}
    There exists a constant $\eps_0>0$ such that the following holds.
    Let $\eps_0>\eps > 0$.
    There is an explicit binary code $C\subseteq\{0,1\}^N$ that
    is decodable from an $\eps$ fraction of insertions/deletions
    with rate $1-\tilde O(\sqrt[4]{\eps})$ in time
    $\poly (N)\cdot(\log N)^{\poly(1/\eps)}$.
  \end{theorem}

\section{High Noise}
  Because our decoding algorithms
  for the $1-\eps$ and $1-\frac{2}{k+1}-\eps$ insertion/deletion
  constructions use the same list decoding technique,
  we abstract out the technical part of the decoding algorithm with
  the following theorem.

  \begin{theorem}
    \label{insertion-deletion-thm}
    Let $C$ be a code over alphabet of size $k$ and length $N=nm$
    obtained by concatenating a Reed-Solomon $C_{out}$ of length $n$
    with an inner code $C_{in}$ of length $m$.
    Suppose $C_{out}$ has rate $r$ and is over $\FF_q$ with $n=q$.
    Suppose $C_{in}:[n]\times\FF_q\to[k]^m$ can correct a
    $1-\delta$ fraction of insertions and deletions in $O(t(n))$ for some function $t$.
    Then, provided $C$ is (combinatorially) decodable under up to $1-\delta-4r^{1/4}$ fraction of insertions
    and deletions, it is in fact decodable in time $O(N^3\cdot(t(N)+\polylog N))$.
  \end{theorem}
  \begin{proof}
    Let $\gamma=4r^{1/4}$.
    Consider the following algorithm, which takes as input
    a string $s$ that is the result of changing a codeword
    $c$ under a fraction $\le (1-\delta-\gamma)$ of insertions/deletions.

    1) $\calJ\leftarrow\emptyset$.

\smallskip
2) For each $0\le j\le \ceil{\frac{2n}{\gamma}}$
      $1\le j'\le\ceil{\frac{4}{\gamma}}$, do the following.
      \begin{enumerate}
        \item[a)] Let $\sigma_{j,j'}$ denote the substring
        from indices $\frac{\gamma m}{2}j$ to $\frac{\gamma m}{2}(j+j')$.
        \item[b)] By brute force search over $\FF_q\times \FF_q$, find all pairs
        $(\alpha,\beta)$ such that 
        $\Delta(\Enc_{C_{in}}((\alpha,\beta)),\sigma_{j,j'})\le (1-\delta) m$.
        If exactly one such pair $(\alpha,\beta)$ exists,
        then add $(\alpha,\beta)$ to $\calJ$.
      \end{enumerate}

3) Find the list, call it $\calL$, of all polynomials
      $p\in\FF_q[X]$ of degree less than $rn$ such that
      $\abs{\{(\alpha,p(\alpha))|\alpha\in\FF_q\}\cap\calJ}\ge \frac{\gamma n}{2}$.

4) Find the unique polynomial in $\calL$, if any,
      such that the insertion/deletion distance between its encoding under $C$ and $s$ is at most $(1-\gamma-\delta)N$.
 
    \textsc{Correctness.}
    Break the codeword $c\in[k]^{nm}$ of the concatenated code $C$ into
    $n$ inner blocks, with the $i$th block $b_i\in[k]^m$
    corresponding to the inner encoding of the $i$th symbol $(\alpha_i,f(\alpha_i))$
    of the outer Reed-Solomon code.
    For some fixed canonical way of forming $s$ out of $c$,
      let $s_i$ be the block formed out of $b_i$, so that
      $s_1,\dots,s_n$ partition the string $s$.
      Call an index $i$ \textit{good} if
      it can be obtained from $b_i$ by at most $(1-\delta-\frac{\gamma}{2})m$
      insertions or deletions, and \textit{bad} otherwise.
    The number of bad indices is at most
    $\frac{(1-\delta-\gamma)mn}{(1-\delta-\gamma/2)m}\le (1-\frac{\gamma}{2})n$,
    so the number of good indices is at least $\frac{\gamma n}{2}$.
    
    For any good index $a$, there exists some $\sigma_{j,j'}$ such that $s_a$
    is a substring of $\sigma_{j,j'}$ and $0<|\sigma_{j,j'}|-|s_a|<\frac{\gamma
    m}{2}$. Since $a$ is good, the insertion/deletion distance
    between $b_a$ and $s_a$
    is at most $(1-\delta-\gamma/2)m$, and the insertion/deletion distance between
    $s_a$ and $\sigma_{j,j'}$ is less than $\gamma m/2$, so the insertion/deletion
    distance between $b_a$ and $\sigma_{j,j'}$ is at most $(1-\delta)m$.
    Since $C_{in}$
    can handle up to $(1-\delta)m$ insertions and deletions,
    it follows that $b_a$ is the unique codeword of $C_{in}$
    such that  $\Delta(\Enc_{C_{in}}(b_a),\sigma_{j,j'})\le (1-\delta) m$.
    Since $b_a$ is the encoding of
    $(\alpha_a,f(\alpha_a))$ under $C_{in}$,
    we conclude that for any good index $a$, the pair
    $(\alpha_a,f(\alpha_a))$ will be included in $\calJ$.
    In particular, $\calJ$ will have at least $\gamma n / 2$
    such pairs, so the correct $f$ will be in $\calL$.

    We now check that step 3 of the algorithm will succeed.
    We have $|\calJ|\le \frac{2n}{\gamma}\cdot \frac{4}{\gamma} = \frac{8n}{\gamma^2}$,
    and Sudan's list decoding algorithm will give
    a list of degree-less-than-$rn$ polynomials over $\FF_q$ such that
    $(\alpha,p(\alpha))\in\calJ$ for more than
    $\sqrt{2(rn)|\calJ|}$ values of $\alpha\in\FF_q$ \cite{sudan1997}.
    Furthermore, this list will have at most $\sqrt{2|\calJ|/(rn)}$
    elements. For our choice of $\gamma$, we have
    $\gamma n/2 > \sqrt{\frac{16rn^2}{\gamma^2}} \ge \sqrt{2(rn)|\calJ|}$,
    so the list decoding will succeed.

    By above, there will be at least one polynomial in $\calL$
    such that the longest common subsequence of its encoding
    with $s$ has length at least $(\gamma+\delta)m$,
    namely the correct polynomial $f$.
    Since we assumed $C$ can decode
    up to a $1-\delta-\gamma$ fraction of insertons/deletions,
    all other polynomials in $\calL$ will have longest common subsequence with $s$
    smaller than $(\gamma+\delta)m$.
    Thus our algorithm returns the correct $f$.

    \textsc{Runtime.} We have $O(n)\le O(N)$
    intervals $\sigma_{j,j'}$ to check, and each one
    brute forces over $n^2$ terms of $\FF_q\times \FF_q$.
    Encoding takes time $O(t(n))\le O(t(N))$ by assumption and
    computing the longest common subsequence takes $O(m^2) = O(\log^2N)$ time,
    so in total the second step of the algorithm takes $O(N^3(t(N)+\log^2 N))$ time.
    Since $|\calJ|\le O(N)$ for sufficiently large $N$,
    the Reed-Solomon list decoding algorithm can be performed in time $O(N^2)$,
    see for instance \cite{RR2000}.
    There are a constant number of polynomials to check at the end,
    and each one takes $O(N^2)$ time using the longest common subsequence algorithm.
    Thus, the overall runtime of the algorithm is $O(N^3(t(N)+\polylog N))$.
  \end{proof}

\subsection{Decoding against $1-\eps$ insertions/deletions} 
  \begin{lemma}
    \label{weakconcatlemma}
    Suppose have a code $C$ which is the concatenation of an outer code
    $C_{out}$ of length $n$ with an inner code $C_{in}$ of length $m$.
    Suppose further that for some $\Delta,\delta\in(0,1)$,
    we have $\LCS(C_{out})\le \Delta n, \LCS(C_{in})\le \delta m$.
    Then $\LCS(C)\le (\Delta+2\delta)nm$.
  \end{lemma}
  \begin{lemma}[$\theta=1/3$ case of Corollary 2.6 of \cite{GW2015}]
    \label{higherrorcorollarylemma}
    Let $1/2>\eps>0$, and $k$ be a positive integer.
    For every $m$, there exists a code $C\subseteq [k]^m$ of rate
    $R=\eps/3$ that can correct a $1-\eps$ fraction of insertions/deletions
    in time $k^{O(m)}$, provided $k\ge 64/\eps^3$.
  \end{lemma}
  \begin{theorem}
    \label{higherrortwotheorem}
    \higherrortwotext
  \end{theorem}
  \begin{proof}
    Let $n=q, m=24\log q/\eps,$ and $k=O(1/\eps^3)$.
    By Lemma \ref{higherrorcorollarylemma}, we can construct by brute force
    a code $C_1:n\times\FF_q\to [k]^m$ that
    can be decoded from $1-\eps/4$ fraction of worst-case insertions and
    deletions.
    We can concatenate $C_1$ with an outer Reed-Solomon code of
    rate $(\eps/8)^4$.

    The rate of the inner code is $\Omega(\eps)$, and the rate of the
    outer code is $\Omega(\eps^4)$, so the total rate is $\Omega(\eps^5)$.

    By Lemma \ref{weakconcatlemma}, $\LCS(C)\le (\eps/8)^4+2(\eps/4)<\eps$,
    so $C$ is capable of decoding up to $1-\eps$ fraction of insertions
    and deletions.
    Encoding in $C_1$ is done by brute force in time $N^{\poly(1/\eps)}$,
    so by Theorem \ref{insertion-deletion-thm}, $C$
    is capable of decoding up to
    $1-\eps/4 - 4((\eps/8)^4)^{1/4} > 1-\eps$ fraction of worst-case
    insertions and deletions in time
    $O(N^3(N^{\poly(1/\eps)}+\poly\log N)) = N^{\poly(1/\eps)}$, as desired.
  \end{proof}
  \begin{remark*}
    Our construction only requires that the inner code can be decoded from $1-\eps/4$
    fraction of worst-case insertions and deletions.
    By using the concatenated code of Theorem \ref{higherrortwotheorem}
    as the inner code of the same construction (thus giving us two
    levels of concatenation), we can reduce the time complexity significantly,
    at the cost of a polynomial reduction in other parameters of the code, as summarized below. 
  \end{remark*}
  \begin{theorem}
    \label{higherrortwofasttheorem}
    For any $\eps>0$, there exists a family of constant rate
    codes over an alphabet
    of size $\poly(1/\eps)$ and rate $\Omega(\eps^9)$
    that can be decoded from a $1-\eps$ fraction
    of insertions/deletions.
    Furthermore, this code is constructible, encodable, and decodable
    in time $\poly (N)\cdot(\log N)^{\poly(1/\eps)}$.
  \end{theorem}
\subsection{Decoding against $1-\frac{2}{k+1}-\eps$ insertions/deletions} 
  First, we sumarize an existence result from \cite{BG2015}.
  \begin{lemma}[Theorem 12 of \cite{BG2015}]
    \label{bg-theorem-12}
    Fix an integer $k\ge 2$ and $\gamma>0$.
    Then there are infinitely many $N$ for which
    there is a concatenated Reed Solomon code $C\subseteq [k]^N$ that
    has outer rate at least $\gamma/2$,
    has total rate at least $(\gamma/k)^{O(\gamma^{-3})}$,
    is decodable under $1-\frac{2}{k+1}-\gamma$ fraction of
      insertions and deletions,
    has an inner code decodable under $1-\frac{2}{k+1}-\gamma/4$
      insertions and deletions,
    and is constructible in time $O(N\log^2 N)$.
  \end{lemma}

  \begin{theorem}
    \label{higherroronetheorem}
    \higherroronetext
  \end{theorem}
  \begin{proof}
    Consider the codes $C$ given by Lemma \ref{bg-theorem-12}
    with $\gamma=2(\eps/5)^4$.
    $C$ has outer rate at least $\gamma/2=(\eps/5)^4$ and total rate
    at least $(\gamma/k)^{O(\gamma^{-3})}$.
    Furthermore,
    $C$ can decode up to $1-\frac{2}{k+1}-\gamma$ fraction of insertions/deletions,
    and the inner code of $C$ can decode $1-\frac{2}{k+1}-\gamma/4$
    fraction of insertions/deletions.
    Thus, by Theorem \ref{insertion-deletion-thm},
    $C$ can efficiently decode up to 
    $1-\frac{2}{k+1}-\gamma/4-4(\gamma/2)^{1/4} > 1-\frac{2}{k+1}-\eps$
    fraction of insertions/deletions.
  \end{proof}


\begin{thebibliography}{1}

\bibitem{SZ1999}
L. Schulman and D. Zuckerman.
\newblock Asymptotically good codes correcting insertions, deletions, and
  transpositions.
\newblock {\em IEEE Trans. Inform. Theory}, 45(7):2552--2557,
  1999.

\bibitem{BG2015}
B. Bukh and V. Guruswami.
\newblock An improved bound on the fraction of correctable deletions.
\newblock {\em Proc. of SODA}, pages 1893-1901, 2016.

\bibitem{GW2015}
V. Guruswami and C. Wang.
\newblock Deletion codes in the high-noise and high-rate regimes.
\newblock {\em Proceeding of RANDOM}, pages 867-880, 2015.

\bibitem{levenshtein1966}
V.~I. Levenshtein.
\newblock Binary codes capable of correcting deletions, insertions, and
  reversals.
\newblock {\em Dokl. Akad. Nauk}, 163(4):845--848, 1965.
\newblock English translation in Soviet Physics Doklady, 10(8):707-710, 1966.

\bibitem{sudan1997}
M. Sudan.
\newblock Decoding of reed solomon codes beyond the error-correction bound.
\newblock {\em J. Complexity}, 13(1):180--193, 1997.

\bibitem{RR2000}
R. M. Roth and G. Ruckenstein.
\newblock Efficient decoding of reed-solomon codes beyond half the minimum
  distance.
\newblock {\em IEEE Trans. Inform. Theory}, 46(1):246--257, 2000.

\end{thebibliography}
\end{document}